\documentclass[graybox]{svmult}

\usepackage{type1cm}        
\usepackage{makeidx}         
\usepackage{graphicx}        
\usepackage{multicol}        
\usepackage[bottom]{footmisc}

\usepackage{booktabs}
\usepackage{caption}

\usepackage{newtxtext}       %
\usepackage{newtxmath}       

\makeindex             

\begin{document}

\title*{Distributed Change Detection via \\ Average Consensus over Networks}
\author{Qinghua Liu, Rui Zhang, and Yao Xie}
\institute{Qinghua Liu \at Department of Electrical Engineering, Princeton University, New Jersey, USA \\ \email{qinghual@princeton.edu}
\and Rui Zhang \at H. Milton Stewart School of Industrial and Systems Engineering, Georgia Institute of Technology, 755 Ferst Drive, NW, Atlanta, GA 30332 \\\email{ruizhang\_ray@gatech.edu}
\and  Yao Xie \at H. Milton Stewart School of Industrial and Systems Engineering, Georgia Institute of Technology, 755 Ferst Drive, NW, Atlanta, GA 30332 \\ \email{yao.xie@isye.gatech.edu}}
%
%
\maketitle
\abstract{Distributed change-point detection has been a fundamental problem when performing real-time monitoring using sensor-networks.  We propose a distributed detection algorithm, where each sensor only exchanges CUSUM statistic with their neighbors based on the average consensus scheme, and an alarm is raised when local consensus statistic exceeds a pre-specified global threshold.  We provide theoretical performance bounds showing that the performance of the fully distributed scheme can match the centralized algorithms under some mild conditions. Numerical experiments demonstrate the good performance of the algorithm especially in detecting asynchronous changes.}

\section{Introduction}
\label{sec:1}
Detecting an abrupt change from data collected by distributed sensors has been a fundamental problem in diverse applications such as cybersecurity \cite{lakhina2004diagnosing, tartakovsky2002efficient} and environmental monitoring \cite{chen2017textsf, valero2017real}. In various applications, it is important to perform {\it distributed detection}, in that sensors perform local decisions rather than having to send all information to a central hub to form a global decision. Some common reasons include (1) local decision at each sensor is needed, such as VANET \cite{li2016order, karagiannis2011vehicular}, where the vehicles need to make immediate decision for traffic condition, by using their own information and by communicating with their neighbors, and (2) limited communication bandwidth, e.g., in distributed geophysical sensor networks \cite{valero2017real} where sensors can only communicate with their neighboring sensors, but cannot communicate to far-away sensors since the channel bandwidth is interference limited, and (3) avoid communicate delay: for seismic early warning systems, it is also not ideal for seismic sensors to send all information to a fusion hub and receiving a global decision, but rather let them to make local decision, to avoid two-way communication delay.


With the above motivation, in this paper, we propose a distributed multi-sensor change-point detection procedure based on average consensus \cite{xiao2004fast}. The scheme lets sensors to exchange their local CUSUM statistics and makes a local decision by comparing their {\it consensus} statistic with a statistic. Note that this scheme does not involve explicit point-to-point message passing or routing; instead, it diffuses information across the network by updating their own statistics by performing a weighted average of neighbors' statistics \cite{xiao2005scheme}. The main theoretical contributions of the paper are the analysis of our detection procedure in terms of the two fundamental performance metrics: the average run length (ARL) which is related to the false alarm rate, and the expected detection delay.  We show that for a system consisting of $N$ sensors, using the average consensus scheme, the expected detecting delay can nearly be reduced by a factor of $N$ compared to a system without communication, under the same false alarm rate. We demonstrate the good performance of our proposed method via numerical examples.

\subsection{Related work}

Various distributed change-point detection methods have been developed based on the classic CUSUM \cite{page1954continuous} and Shiryaev-Roberts statistics. Many existing distributed methods \cite{tartakovsky2002efficient, tartakovsky2003quickest, tartakovsky2008asymptotically, mei2010efficient} assume a fusion center that gathers information (raw data or statistics) from all sensors to perform decision globally. Thus, they are different from our approach where each sensor performs a local decision. 
On the other hand, there is another type of approaches such as the ``one-shot'' scheme, where each sensor makes a decision using its own data and only transmits a one-bit signal to central hub once a local alarm has been trigged (e.g., \cite{hadjiliadis2009one,tartakovsky2008asymptotically}). However, this approach can be improved if the change is observed by more than one sensor, and we can allow neighboring sensors to exchange information.  Fig. \ref{fig0} illustrates a comparison of our approach versus the other two types of approaches. 


\begin{figure}
	\centering
	\begin{minipage}[c]{0.3\textwidth}
		\centering
\includegraphics[width = \textwidth]{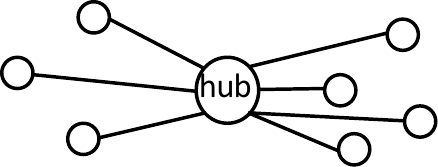}
	\end{minipage}
	\begin{minipage}[c]{0.3\textwidth}
		\centering
\includegraphics[width = \textwidth]{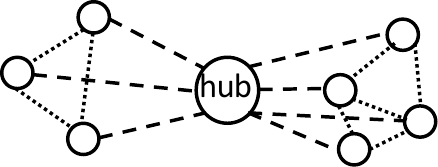}
	\end{minipage}
	\begin{minipage}[c]{0.3\textwidth}
		\centering
\includegraphics[width = \textwidth]{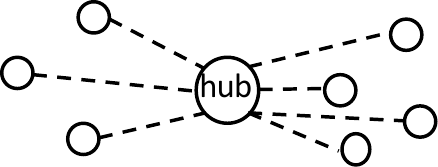}
	\end{minipage}
\caption{Comparison of centralized approach (left), our approach (middle) and one-shot scheme (right). Solid line: communication of raw data or statistics. Dash line: communication of one-bit decisions. Dash dots: communication of statistics.}
\label{fig0}
\vspace{-0.1in}
\end{figure}

Some recent works \cite{li2016order,SahuKar2016,LiuMei2017} study a related but different problem: distributed sequential hypothesis test based on average consensus. A major difference, though, is that in the sequential hypothesis test, the local log-likelihood statistic accumulates linearly, while in sequential change-point detection, the local detection statistic accumulates nonlinearly as a reflected process (through CUSUM). This results in a more challenging case and requires significantly different techniques. 

Moreover, recent works \cite{raghavan2010quickest, ludkovski2012bayesian, fellouris2016second, kurt2017multi} study the model under the general setting where not all the nodes have a change point or have different change points; \cite{kurt2017multi, raghavan2010quickest,ludkovski2012bayesian} assume the influence from the source propagates to each sensor nodes sequentially under some prior distribution.  Here we do not make an assumption about how the change is observed by different sensors. 


\subsection{Background}

We first introduce some necessary notations.
Given two distinct distributions $\mathcal{P}_1$  and $\mathcal{P}_2$. Let the probability density function of $\mathcal{P}_1$ and $\mathcal{P}_2$ be $f_1({\bf x})$ and $f_2({\bf x})$, respectively.  Then the log-likelihood ratio function (LLR) between distribution $\mathcal{P}_2$ and $\mathcal{P}_1$ is defined as
$
L({\bf x}) = \log [f_2({\bf x})/f_1({\bf x})].
$

Assume a sequence of observations $\{ {\bf x}^{t}\}_{t=1}^{+\infty}$.  There may exist a change-point $\tau$, such that for $t < \tau$, ${\bf x}^t\stackrel{i.i.d.}\sim \mathcal{P}_1$ and for $t\ge \tau$, ${\bf x}^t\stackrel{i.i.d.}\sim \mathcal{P}_2$. 
%
The classical CUSUM procedure is based on the LLR to detect the change of the data distribution. It is a stopping time that stops the first time the LLR based statistic exceeds a threshold $b$:
$
T_s = \inf\big\{t>0: \max_{1\le i \le t}\sum_{k=i}^{t}L({\bf x}^k)\ge b \big\}.
$
The stopping time $T_s$ has a recursive implementation:
$
y^{t+1}  = \max \{y^t+L ({\bf x}^{t+1}),0\}, y^0 = 0,$ and 
$T_s = \inf\big\{t>0: y^t\ge b \big\}. $


\section{Distributed consensus detection procedure}

We represent an $N$-sensor network using a graph $G = (\mathcal{V},\mathcal{E})$, where $\mathcal{V}$ and $\mathcal{E}$ are the sensor set and edge set, respectively.
There exists an edge between sensor $i$ and sensor $j$ if and only if they can communicate with each other.
Without loss of generality, we assume that the $G$ is connected (if there is more than one connected component, we can apply our algorithm to each of them separately.) Assume the topology of the sensor network is known (e.g., by design). 

Denote data observed by the sensor $v$ at time $t$ as ${\bf x}_{v}^{t}$. Consider the following change-point detection problem. When there is no change, the sensor observations ${\bf x}_v^t \stackrel{i.i.d.}\sim \mathcal{P}_{1},\ \forall v, t = 1, 2, \ldots$.  When there is a change, at least one sensor will be affected a change that happens at an unknown time $\tau$, such as $\ {\bf x}_v^1,\cdots,{\bf x}^{\tau-1}_v	\stackrel{\rm i.i.d.}\sim \mathcal{P}_{1},$  and ${\bf x}_v^{\tau},\cdots,{\bf x}^{T}_v	\stackrel{\rm i.i.d.}\sim \mathcal{P}_{2}$. Our goal is to detect the change as quickly as possible (for at least one sensor that has been affected by the change), subject to the false-alarm constraint. 


 
Our {\it distributed consensus change-point detection procedure} consists of three steps at each sensor: 
(1) Each sensor forms local CUSUM statistic using their own data:  
$
y_v^{t+1}  = \max \{y_v^t+L_v ({\bf x}_v^{t+1}),0\}, \   v\in\mathcal{V}; 
$
(2) Sensors exchange information with their neighbors according to the pre-determined network topology and weights to form the {\it consensus} statistic: 
$
z_v^{t+1} = \sum_{u \in \mathcal{N}(v)} W_{vu}\left(z_u^t + y_u^{t+1}- y_u^{t}   \right),\  v\in\mathcal{V}, $
where $\mathcal{N}(v)$ includes sensor $v$ and its neighbors.  
(3) Perform detection by comparing $z_v^t$ with a predetermined threshold $b$ at each sensor $v\in\mathcal{V}$.
If a global decision is necessary, as long as there exists one sensor $v\in\mathcal{V}$ that raises an alarm: $z_v^t\ge b$ a global alarm is raised.
In summary, our detection procedure corresponds to the following stopping time 
\begin{equation}\label{stoptime}
T_s = \inf\big\{t>0: \max_{v\in\mathcal{V}} z_v^t \ge b\big\}.
\end{equation}

We assume the weighted consensus matrix ${\bf W} \in \mathbb{R}^{N \times N}$, which the sensors use to exchange information, will satisfy the following conditions.  As long as the graph is connected, the consensus matrix ${\bf W}$ satisfying the above conditions always exists \cite{boyd2004fastest}. 
\begin{enumerate}
\item[(i)] $W_{ij}>0$ if sensor $i$ and sensor $j$  are connected and $W_{ij}=0$ if sensor  $i$ and sensor $j$  are not connected.
\item[(ii)] Assume communication in the network is symmetrical, i.e. $W_{ij} = W_{ji}$; this happens when sensors broadcast to their neighbors. 
\item[(iii)] ${\bf W}{\bf 1} = {\bf 1}$, meaning that the information is not augmented or shrunken during communication, where {\bf 1} is the all-one vector. 
\item[(iv)] The second largest eigenvalue modulus of the matrix $\lambda_{2}({\bf W})$ is smaller than $1$ (to ensure convergence of the algorithm). 
\end{enumerate}

%

\section{Theoretical analysis of ARL and EDD}

We now present the main theoretical results.  We adopt the standard performance metrics for sequence change-point detection: the average run length (ARL) and the expected detection delay (EDD) \cite{xie2013sequential}, defined as
${\rm ARL}  =\mathbb{E}[T_s|\tau = \infty]$, and ${\rm  EDD} = \mathbb{E}[T_s|\tau = 1]$ (assuming the change occurs that the first moment, for simplicity).  
In the definition above, $\tau = \infty$ means that the change-point never occurs. 
Intuitively, ${\rm EDD}$ can be interpreted as the delay time before detecting the change and ${\rm ARL}$ can be interpreted as the expected duration between two false alarms. 
We make the following assumptions
\begin{enumerate}
\item[(1)] All the sensors share the same pre- and post-change distributions $\mathcal{P}_1$ and $\mathcal{P}_2$ (if the change occurs). 
\item[(2)]  For ${\bf x} \sim \mathcal{P}_1$ and ${\bf x} \sim \mathcal{P}_2$, random  $L({\bf x})$ follows a non-central sub-Gaussian distribution \cite{buldygin1980sub}.  Assumption for LLR to be a non-central sub-gaussian distribution can capture many commonly seen cases. For instance, 
 Gaussian distributions $\mathcal{P}_1 = \mathcal{N}({\bf 0},{\bf I})$,  $\mathcal{P}_2 = \mathcal{N}({\bf u},{\bf I})$ lead to $L({\bf x}) = {\bf u}^{\rm T}{\bf x}-\Vert {\bf u}\Vert^2/2$, which follows
$L({\bf x})\sim \mathcal{N}(\Vert {\bf u}\Vert^2/2,\Vert {\bf u}\Vert^2)$.
  \end{enumerate}
The above assumption is made purely for theoretical analysis. The detection procedure can still  be implemented without these assumptions.

First we present an asymptotic lower bound for the ARL. 
Assume that the mean and variance of $L({\bf x})$ when  ${\bf x} \sim \mathcal{P}_1$  are given by $\mu_1$ and $\sigma_1$, respectively.
Note that  $(-\mu_1$) corresponds to the Kullback-Leibler (KL)-divergence from $\mathcal{P}_2$ to $\mathcal{P}_1$, and $(-\mu_1) \geq 0$  always holds, which can be shown using Jensen's inequality.

\begin{theorem}[{\bf Lower-bound for ARL}]\label{thm1}
When $b\rightarrow \infty$, we have
 \begin{align*}
& {\rm ARL} \ge  \exp\bigg\{ \frac{(-\mu_1) b}{\sigma_1^{2}}\left[2N-2\left(\frac{N}{N+1}\right)^2\right] \\
&\quad + \left[\frac{ \sqrt{-\mu_1}\mu_1\lambda_{2}({\bf W})}{\sigma_1^2\left[1-\lambda_{2}({\bf W})\right]}\left(4N^2-4N\left(\frac{N}{N+1}\right)^2\right)+o(1)\right]\sqrt{b}\bigg\}.
 \end{align*}
\end{theorem}
The  theorem shows that the ARL increases exponentially with threshold $b$, which is a desired property of a detection procedure. 
Moreover, it shows that it increases at least exponentially as $N$ increases.
The detailed proof is delegated to appendix.

%

Now we present an asymptotic lower bound to EDD. 
Denote the mean and the variance of $L({\bf x})$ for ${\bf x} \sim \mathcal{P}_2$  as $\mu_2$ and $\sigma_2$, respectively.
Note that $\mu_2 \geq 0$ corresponds to the KL-divergence from $\mathcal{P}_1$ to $\mathcal{P}_2$.
\begin{theorem}[Upper-bound for EDD]\label{thm2}
When $b\rightarrow\infty$, we have
\begin{align*}
{\rm EDD}\le \frac{b}{\mu_2}\left(1+o(1)\right).
\end{align*}
\end{theorem}
Comparing the upper bound with the lower bound in \cite{lorden1971procedures}, we may be able to show that the proposed procedure is first-order asymptotically optimal (which is omitted here due to space limit). 
Moreover, combining Theorem \ref{thm1} with Lemma \ref{thm2}, we can characterize the relationship between ARL and  EDD as follows
\begin{corollary}[ARL and EDD]\label{arl-edd}
When $b\rightarrow\infty$, if ${\rm ARL} \ge \gamma$, we have
\begin{align*}
{\rm EDD} \le \frac{\log\gamma \left(1+o(1)\right)}{N\mu_2} \times \frac{\sigma_1^2}{-2\mu_1\left(1-N/(N+1)^2\right)}.
\end{align*}
\end{corollary}
Corollary \ref{arl-edd} shows that the ratio between the EDD of our algorithm and that of the one-shot scheme \cite{hadjiliadis2009one} is no larger than $\sigma_1^2/ [-2N\mu_1\left(1-N/(N+1)^2\right)]$. Similarly, by comparing Theorem \ref{arl-edd} with the results in \cite{tartakovsky2003quickest,mei2010efficient}, the ratio is no larger than $\sigma_1^2/ [-2\mu_1\left(1-N/(N+1)^2\right)]$.

\section{Numerical Experiments}

In this section, we present several numerical experiments to demonstrate the performance of our algorithm.  Assume  $\mathcal{P}_1$ is  $\mathcal{N}(0,1)$ and $\mathcal{P}_2$ is $\mathcal{N}(1,1)$. Thus, $\mu_1 = -0.5$, $\mu_2 = 0.5$ and $\sigma_1=\sigma_2=1$.
%
%
We consider a simple network with $N = 4$ for illustrative purposes. Consider two network topology, a line network (where sensors communicate with their neighbors) and K4 (a fully connected network, which can be viewed as unrealistic upper bound for performance). The second largest eigenvalue modulus for line network and K4 are 0.9 and 0, respectively. Their weight matrices are given by
\begin{align*}
\mbox{Line:\ }
 \left(
 \begin{matrix}
    5/8  & 3/8  & 0 & 0  \\
  3/8 & 1/2 & 1/8 & 0 \\
   0 & 1/8 & 1/2& 3/8 \\
   0 & 0 & 3/8 &5/8
  \end{matrix}
  \right)
   \quad
  \mbox{K4:\ }
 \left(
 \begin{matrix}
    1/4  & 1/4  & 1/4  & 1/4  \\
1/4  & 1/4  & 1/4  & 1/4  \\
1/4  & 1/4  & 1/4  & 1/4  \\
1/4  & 1/4  & 1/4  & 1/4
  \end{matrix}
  \right)
  \end{align*}
We compare the the performance of our proposed procedure with the one-shot scheme \cite{hadjiliadis2009one} and the centralized approach where the sum of all local CUSUM statistics is compared with a threshold. We calibrate the threshold of all approaches by simulation, so that they will have the same ARL when there is no change, to have a fair comparison.

{\bf Synchronous changes}. In the first experiment, we assume the change-point happens at the same time at all sensors.   The results are presented in Fig. \ref{fig_sameChange}.  We find that the performance of K4 and centralized approach are the same in this case, since all sensor information are used. Since the change-point happens at all sensors synchronously, the one-shot scheme is least favored because each sensor works alone and did not utilize information at other sensors.  


\begin{figure}[h]\label{fig3.0}
	\subfigures
	\centering
	\begin{minipage}{0.45\linewidth}
		\centering
		\includegraphics[width =\linewidth]{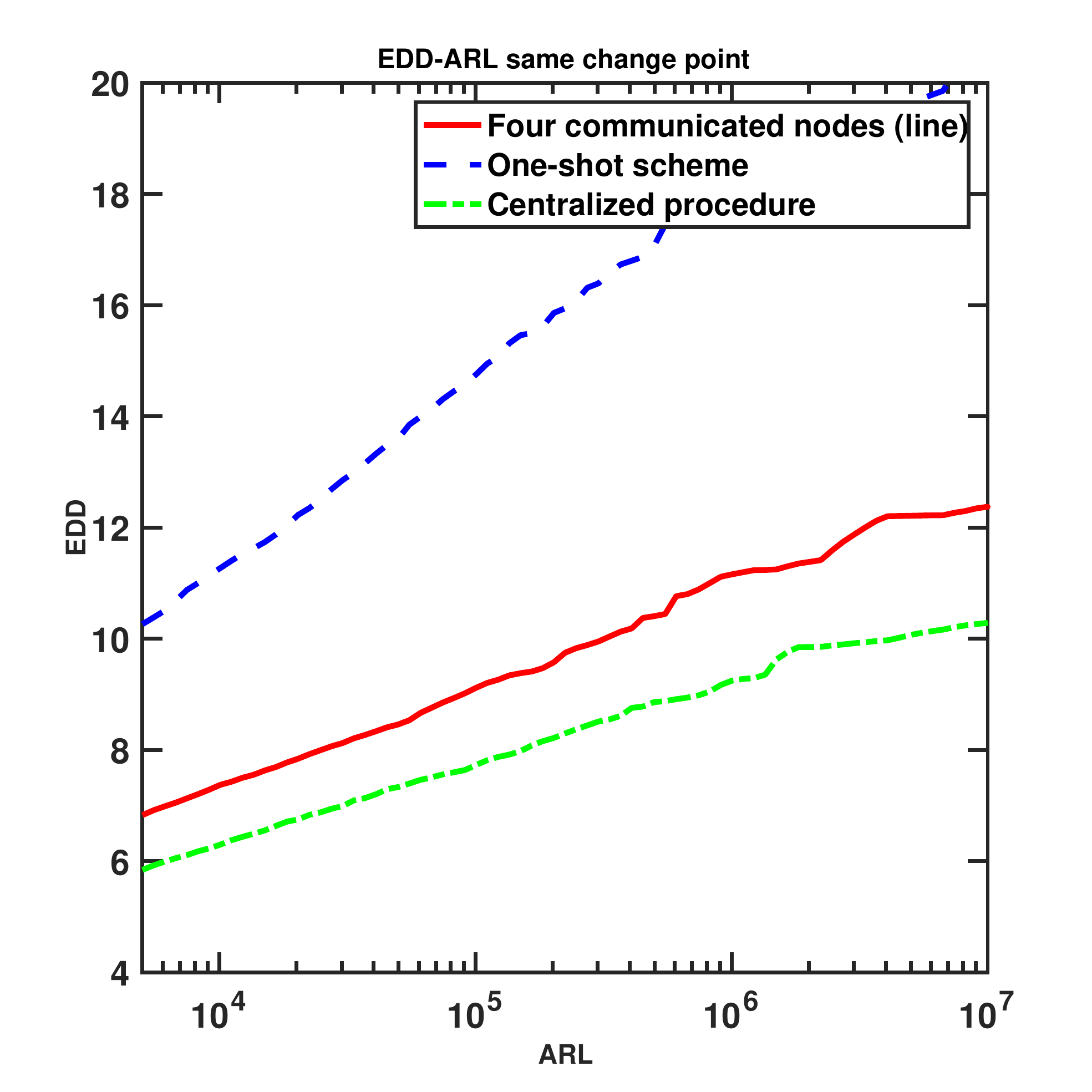}
		\caption{}
		\label{fig_sameChange}
	\end{minipage}
	\begin{minipage}[h]{0.45\textwidth}
		\centering
		\includegraphics[width =\textwidth]{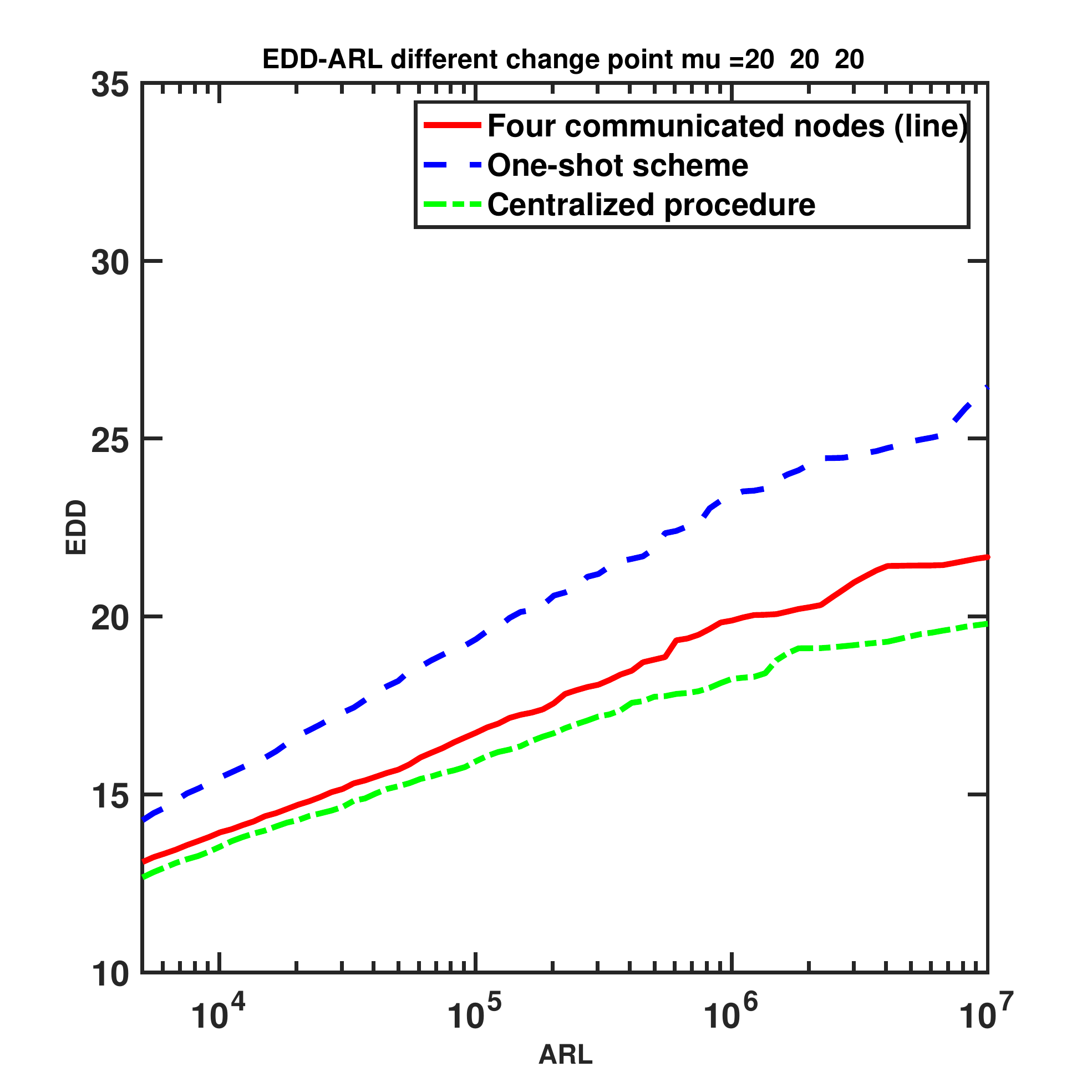}
		\caption{}
		\label{fig3.1}
	\end{minipage}
	
	\begin{minipage}[c]{0.45\textwidth}
		\centering
		\includegraphics[width =\textwidth]{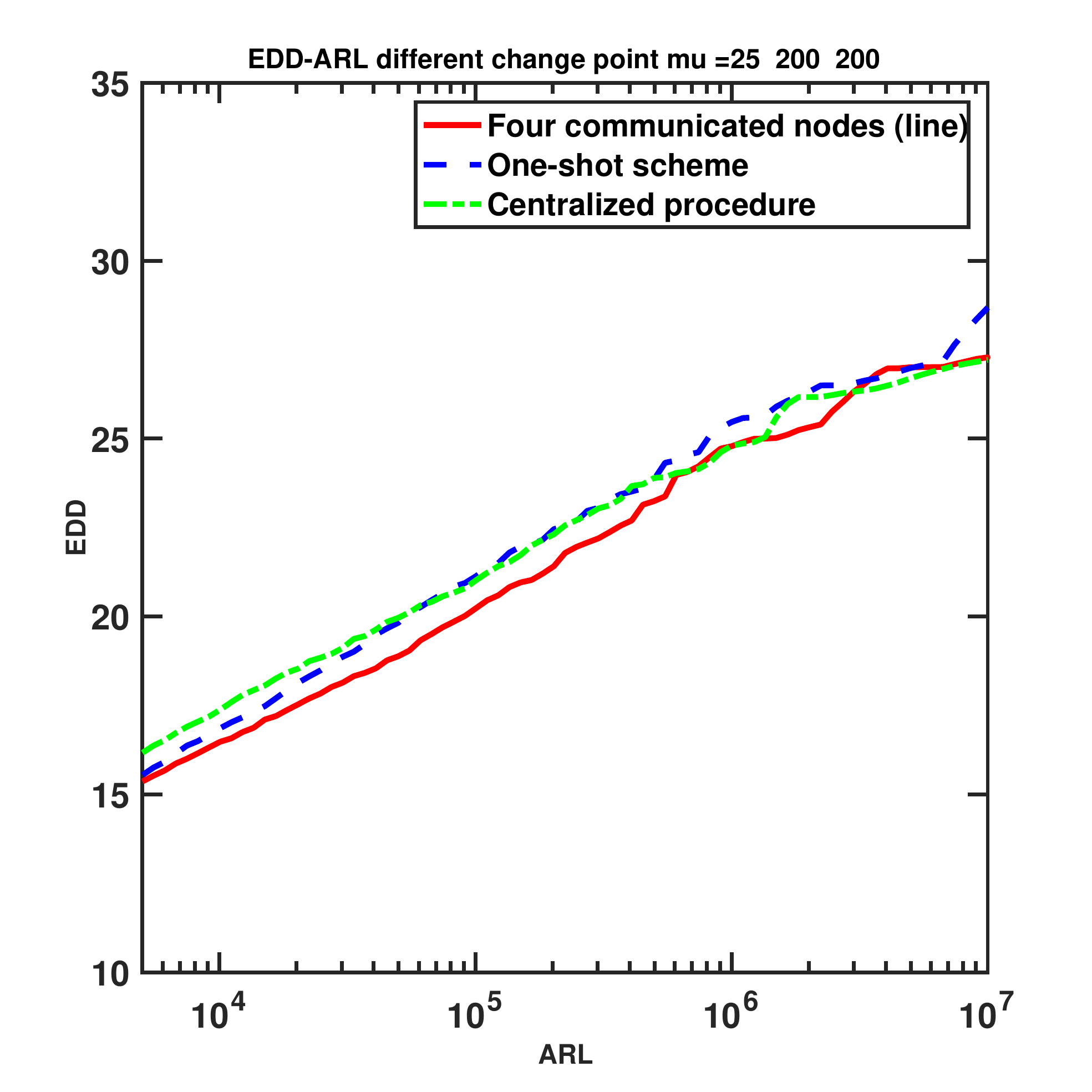}
		\caption{}
		\label{fig3.2}
	\end{minipage}
	\begin{minipage}[c]{0.45\textwidth}
		\centering
		\includegraphics[width =\textwidth]{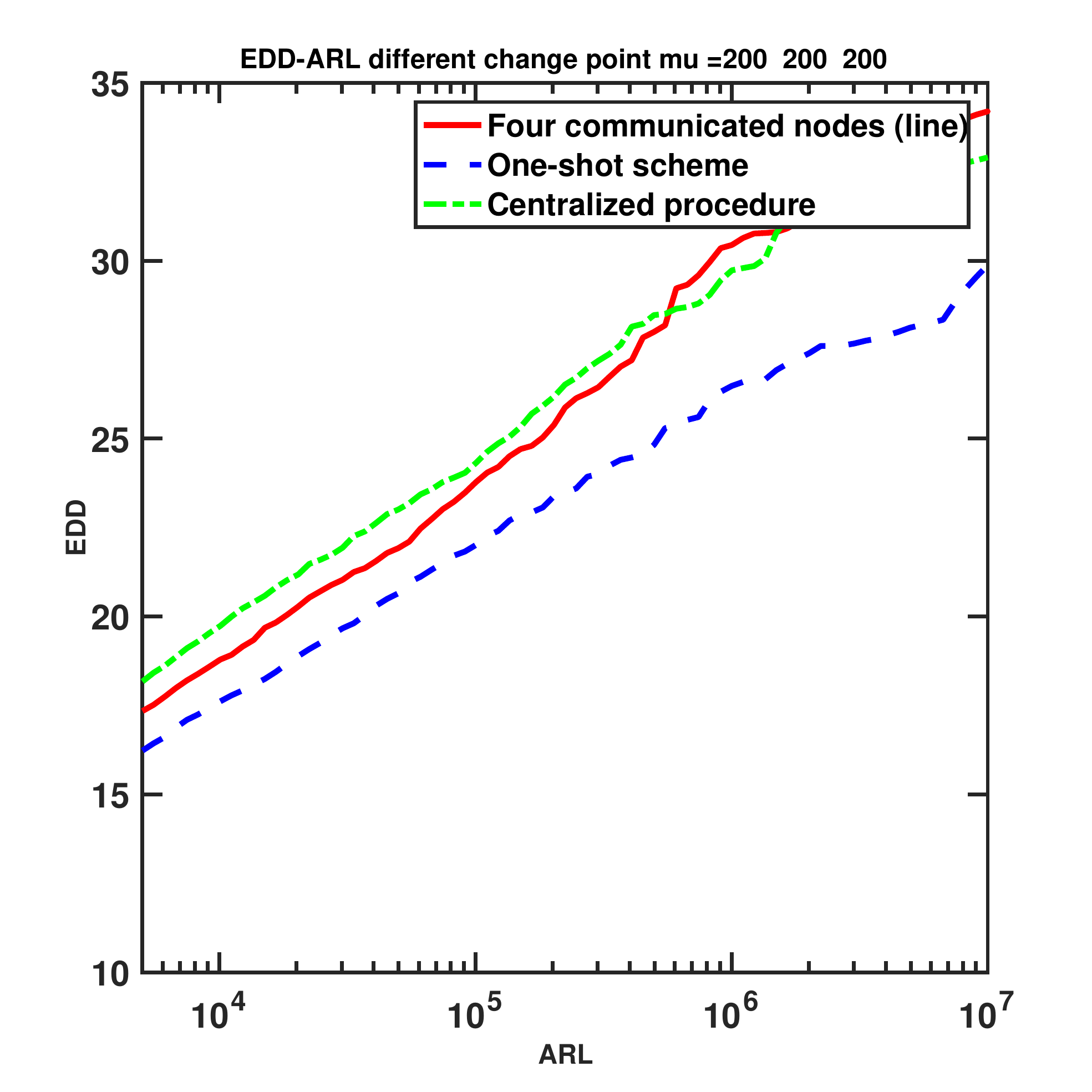}
		\caption{}
		\label{fig3.3}
	\end{minipage}
\end{figure}
\begin{figure}
\samenumber
\caption{Performance of our procedure, one-shot and centralized procedure. In \ref{fig_sameChange}, change-point happens at all sensors at the same time. In \ref{fig3.1}, \ref{fig3.2}, \ref{fig3.3}, change-points happens at all sensors with a random exponential delay.  When simulating EDD, we set $\tau_1 = 1$. In Fig. \ref{fig_sameChange}: $\tau_2,\tau_3,\tau_4 =\tau_1$, in Fig. \ref{fig3.1}: $\tau_2,\tau_3,\tau_4 \sim \mbox{Exp}(20)$; in Fig. \ref{fig3.2}: $\tau_2\sim \mbox{Exp}(25),\tau_3,\tau_4 \sim \mbox{Exp}(200)$; \ref{fig3.3}: $\tau_2,\tau_3,\tau_4 \sim \mbox{Exp}(200)$.}	
\end{figure}
%
%

{\bf Asynchronous changes.} The benefit of our proposed procedure is more significant in the asynchronous case, i.e., when the change-point happens at affected sensors at a different time. In this experiment, we consider three cases: (1) the change-point observed at sensors with random delay in a small range, (2) two sensors observe the change-point with random delay in a small range, and others with random delay in a larger range, and (3) all sensors experience a large range of random delay. Fig. \ref{fig3.1} shows that in Case (1), the centralized approach is the best which is similar to the synchronous change-point case. Fig. \ref{fig3.3} shows Case (3), the one-shot scheme is the best since the changes observed at different sensors may be far apart in time and less helpful in making a consensus decision. Fig. \ref{fig3.2} shows that in Case (2), our proposed procedure can be better than both the one-shot and centralized procedures. This shows that when there is a reasonable delay between changes at different sensors, the consensus algorithm may be the best approach. 


\begin{figure}[h]
\centering
\includegraphics[width =1\linewidth]{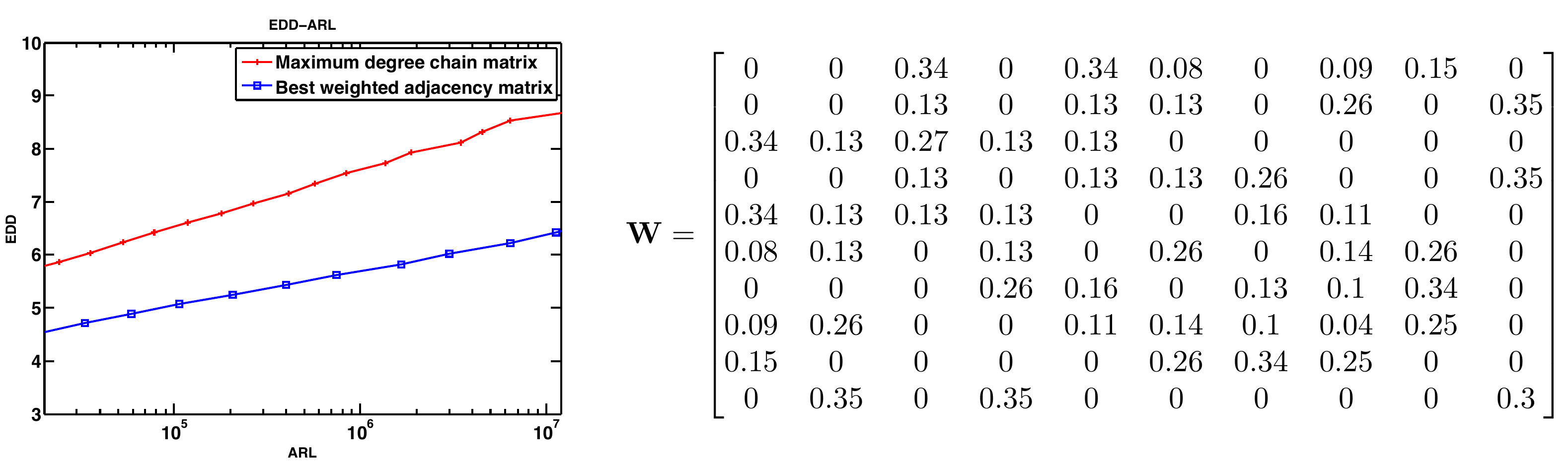}
\caption{Optimized consensus matrix \cite{boyd2004fastest} versus Maximum degree chain matrix}
\label{fig5}
\end{figure}

{\bf Optimize consensus weights.} To demonstrate the effect of consensus weights, we compare two networks with the same topology: the first network is the maximum degree chain network \cite{boyd2004fastest}, which uses unity weights on all edges, and the second network uses optimized weights, which are obtained using the algorithm in \cite{boyd2004fastest} for a fixed topology by minimizing the second largest eigenvalue modulus to achieve faster convergence.  We test the performance of our algorithm on the optimal consensus matrix and another type of consensus matrix, the maximum degree chain matrix. The topology and optimal weights used in this experiment is shown in Fig. \ref{fig5}.
%
The maximum degree chain network has the following weights 
\begin{eqnarray}
\begin{small}
W_{ij}=
\begin{cases}
1/\max_{i \in \mathcal{V}} d_i, &i\neq j \mbox{ and } (i,j)\in \mathcal{E}\cr   1-\sum_{j\in \mathcal{N}(i)} W_{ij}, &i=j \cr 0, &\mbox{otherwise}, \end{cases}
\end{small}
\end{eqnarray}
where $d_i$ is the number of neighbors of sensor $i$. Their second largest eigenvalue modulus are $0.5722$ (optimized weights) and $0.7332$ (maximum degree network), respectively.
Fig. \ref{fig5} shows that the optimized consensus matrix achieves certain performance gain by optimizing weights for the same network topology, which is consistent with Theorem \ref{thm1}. This example shows that when fixing the network topology (which corresponds to fixing the support of ${\bf W}$, i.e., the location of the non-zeros), there are still gains in optimizing the weights to achieve better performance. 

\section{Conclusion}
In this paper, we present a new distributed change-point detection algorithm based on average consensus, where sensors can exchange CUSUM statistic with their neighbors and perform local detection. Our proposed procedure has low communication complexity and can achieve local detection. We show by numerical examples that by allowing sensors to communicate and share information with their neighbors, the sensors can be more effective in detecting asynchronous change-point locally.

\section*{Acknowledgement}

We would like to thank Professor Ansgar Steland for the opportunity to submit an invited paper. This work was partially supported by NSF grants CCF-1442635, CMMI-1538746, DMS-1830210, an NSF CAREER Award CCF-1650913, and a S.F.  Express award. 

 \bibliographystyle{spmpsci.bst}
 \bibliography{reference.bib}


\section*{Appendix}
For  simplicity, we first inded the sensors from $1$ to $N$.
Use vector ${\bf L}^t$ to represent $\left(L({\bf x}_1^{t}),\cdots,L({\bf x}_N^{t})\right)^{\rm T}$, vector ${\bf y}^t$ to represent $\left(y_1^t,\cdots,y_N^t\right)^{\rm T}$
and vector ${\bf z}^t$ to represent $\left(z_1^t,\cdots,z_N^t\right)^{\rm T}$. Now, our algorithm can be rewritten as
\begin{align}\label{rule}
&{\bf y}^{t+1} = ( {\bf y}^t + {\bf L}^{t+1})^+, {\bf z}^{t+1} = {\bf W}({\bf z}^t + {\bf y}^{t+1}-{\bf y}^t), 
&T_s = \inf\big\{t>0: \Vert {\bf z}^t \Vert_{\infty} \ge b\big\}.
\end{align}
Firstly, we prove some useful lemmas before reaching the main results.

\begin{remark}
Since ${\bf W}{\bf 1} = {\bf 1}$, ${\bf W}^{\rm T}={\bf W}$ and $z^0_v = y^0_v=0$,  simple proof by m.i. can verify
\begin{align} \label{sumz_y}
\sum_{v\in\mathcal{V}} z_v^t = \sum_{v\in\mathcal{V}} y_v^t \mbox{ holds for all } t,
\end{align}
\end{remark}
\begin{lemma}({\bf Hoeffding Inequality})
Let $X_i$ be independent, mean-zero, $\sigma_i^2$-sub-Gaussian random variables. Then for $K>0$,
$
\mathbb{P}(\sum_{i=1}^{n} X_{n}\ge K)\le  \exp \left(-\frac{K^2}{2\sum_{i=1}^{n}\sigma_i^2}\right).
$
\end{lemma}
\begin{lemma}\label{lemma_3}
  Consider a sequence of random variables $X_k \stackrel{i.i.d.}\sim \mathcal{P}$, for $k=1,2,\ldots,t$. $\mathcal{P}$ is a sub-Gaussian distribution and its mean and variance are defined as $\mu_1<0$ and $\sigma_1$, respectively. Given $K>0$ large enough, we have
\begin{align*}
\sum_{k=1}^{t}\mathbb{P}\left(\sum_{q=1}^{k} X_{q} > K\right) <  -\frac{2K}{\mu_1} \exp\left(\frac{2K\mu_1}{\sigma_1^{2}}\right).
\end{align*}
\end{lemma}

\begin{proof} 
{\bf Case 1}. For $0 < t \le [-\frac{2K}{\mu_1}]$, by Hoeffding Inequality, we have
\begin{align}\label{eq2}
\sum_{k=1}^{t}\mathbb{P}\left(\sum_{q=1}^{k} X_{q} > K\right)
< \sum_{k=1}^{t}  {\rm exp}\left(-\frac{1}{2}(\frac{K-k\mu_1}{\sqrt{k}\sigma_1})^{2}\right).
\end{align}
Using $\frac{K-k\mu_1}{\sqrt{k}} \ge 2\sqrt{-K\mu_1}$ and $t \le -\frac{2K}{\mu_1}$, we obtain
\begin{align}\label{eq4}
 \sum_{k=1}^{t}\mathbb{P}\left(\sum_{q=1}^{k} X_{q} > K\right) < -\frac{2K}{\mu_1} {\rm exp}\left(\frac{2K\mu_1}{\sigma_1^{2}}\right).
\end{align}
{\bf Case 2.} For $[-\frac{2K}{\mu_1}]+1 \le t $, by \eqref{eq4}, we have
\begin{align}\label{eq51}
 \sum_{k=1}^{t}\mathbb{P}\left(\sum_{q=1}^{k} X_{q} > K\right) 
< -\frac{2K}{\mu_1} {\rm exp}\left(\frac{2K\mu_1}{\sigma_1^{2}}\right)+ \sum_{k=[-\frac{2K}{\mu_1}]+1}^{t}\mathbb{P}\left(\sum_{q=1}^{k} X_{q} > K\right).
\end{align}
Utilizing Hoeffding Inequality and  $k\ge[-\frac{2K}{\mu_1}]+1$, we obtain
\begin{align}\label{eq60}
 \mathbb{P}\left(\sum_{q=1}^{k} X_{q} > K\right)<{\rm exp}\left(-\frac{1}{2}(\frac{K-k\mu_1}{\sqrt{k}\sigma_1})^{2}\right)
\le  {\rm exp}\left(\frac{9K\mu_1}{4\sigma_1^2}\right).
\end{align}
Besides,  for $k \ge [-\frac{2K}{\mu_1}]+1$ ,we have
\begin{align}\label{eq6}
 \frac{ {\rm exp}\left(-\frac{1}{2}(\frac{K-(k+1)\mu_1}{\sqrt{k+1}\sigma_1})^{2}\right) }{ {\rm exp}\left(-\frac{1}{2}(\frac{K-k\mu_1}{\sqrt{k}\sigma_1})^{2}\right)}
= {\rm exp}\left(-\frac{\mu_1^2}{2\sigma_1^2}+\frac{K^2}{2k(k+1)\sigma_1^2}\right)
< {\rm exp}\left(-\frac{3\mu_1^2}{8\sigma_1^2}\right).
\end{align}
Then, from Hoeffding Inequality, \eqref{eq60} and \eqref{eq6},  we derive
\begin{align}\label{eq71}
& \sum_{k=[-\frac{2K}{\mu_1}]+1}^{t}\mathbb{P}\left(\sum_{q=1}^{k} X_{q} > K\right) 
 < \sum_{k=[-\frac{2K}{\mu_1}]+1}^{t}{\rm exp}\left(-\frac{1}{2}(\frac{K-k\mu_1}{\sqrt{k}\sigma_1})^{2}\right)\\
 <& \sum_{k=[-\frac{2K}{\mu_1}]+1}^{t}    {\rm exp}\left(\frac{9K\mu_1}{4\sigma_1^2}\right) \times {\rm exp}\left(-\frac{3\mu_1^2}{8\sigma_1^2} \left(k-[-\frac{2K}{\mu_1}]-1\right)\right)
<\frac{ {\rm exp}\left(\frac{9K\mu_1}{4\sigma_1^2}\right)}{1-{\rm exp}\left(-\frac{3\mu_1^2}{8\sigma_1^2}\right)}.\nonumber
\end{align}
From \eqref{eq71}, we know that the second term on the RHS of  \eqref{eq51} is a small quantity compared with the first term provided $K$ large enough, so we can neglect it to obtain
\begin{align}\label{eq5}
\sum_{k=1}^{t}\mathbb{P}\left(\sum_{q=1}^{k} X_{q} > K\right) <  -\frac{2K}{\mu_1} {\rm exp}\big(\frac{2K\mu_1}{\sigma_1^{2}}\big).
\end{align}
\end{proof}
Note that
$
\mathbb E_{f_1} [L(x_j^t)] = \mu_1 < 0
$, 
$
\mathbb E_{f_2} [L(x_j^t)] = \mu_2 > 0
$, 
$
\mbox{Var}_{f_1} [L(x_j^t)] = \sigma_1^2 
$, 
$
\mbox{Var}_{f_2} [L(x_j^t)] = \sigma_2^2 
$.

Given $\varepsilon >0$ and $p>0$, Define event
$$B(\varepsilon,p)= \{ |L({\bf x}_i^t)|< \varepsilon b,\mbox{ for }i=1,\ldots,N\mbox{ and }t=1,\ldots,p \},$$ 
where $b$ is the pre-specified threshold in detection. Besides, we use $\{T_s = p\}$ to represent the event that our algorithm detects the change at $t=p$.  We have the following lemma
\begin{lemma}\label{lemma_1}
For any $t \le p $, we have
 \begin{align*}
 \{T_s = t\}\wedge   B(\varepsilon,p)
 \subset   \big\{  \frac{\sum _{j=1}^{N} y_{j}^{t}}{N}> (1- \frac{ \sqrt{N}\varepsilon  \lambda_{2} }{1- \lambda_{2}})b \big\} \wedge B(\varepsilon,p).
 \end{align*}
\end{lemma}
\begin{proof}
Note that
$
{\bf W} = \frac 1 N {\bf 1} {\bf 1}^\intercal + \sum_{j=2}^N \lambda_j u_j u_j^\intercal.
$
Throughout the proof, we assume under the condition that $B(\varepsilon,p)$ occurs. First, by the recursive form of our algorithm in \eqref{rule}, the result in \eqref{sumz_y} and the definition of  $B(\varepsilon,p)$, for any sensor $j$, we have
\begin{align*}
 &|z_{j}^{t} - \frac{\sum _{i=1}^{N} y_{i}^{t}}{N} | = | z_{j}^{t} - \frac{\sum_{i=1}^{N} z_{i}^{t}}{N} | 
 \le \Vert {\bf z}^t - \frac{\sum_{i=1}^{N} z_{i}^{t}}{N} {\bf 1}\Vert_2 
  = \Vert \sum_{k=1}^{t} ({\bf W}^{t-k+1}-\frac{1}{N}{\bf 1}{\bf 1}^{\rm T})({\bf y}^{k}-{\bf y}^{k-1}) \Vert_{2}\\
 \le&   \sum_{k=1}^{t}  \lambda_{2}^{t-k+1} \Vert{\bf y}^{k}-{\bf y}^{k-1} \Vert_{2} 
 \le   \sum_{k=1}^{t}  \lambda_{2}^{t-k+1} \Vert{\bf L}^{k} \Vert_{2} 
 \le   \sum_{k=1}^{t}  \lambda_{2}^{t-k+1} \sqrt{N}\varepsilon b
\le  \frac{  \sqrt{N}\varepsilon  \lambda_{2}b}{1- \lambda_{2}},
\end{align*}
where $\lambda_2$ is the second largest eigenvalue modulus of ${\bf W}$. If $\{T_s = t\}$ happens, then $z_{j}^{t}>b$ holds for some $j$, which, together with the inequality above, leads to
$
\frac{\sum _{j=1}^{N} y_{j}^{t}}{N}> (1- \frac{  \sqrt{N}\varepsilon  \lambda_{2}}{1- \lambda_{2}})b.
$
\end{proof}

\begin{lemma}\label{lemma_2}
Assume a sequence of independent random variables $Y_1,\cdots,Y_N$. Take any integer $M>N$ and let
\begin{align*}
C(M,N) \big\{(i_1,\cdots,i_N): i_j \in \mathbb{N} \ \mbox{and}\  M-N \le \sum_{j=1}^{N}i_{j} \le M \big\}.
 \end{align*}
 Then we have
\begin{align*}
\mathbb{P}\left(\sum_{j=1}^{N} Y_j > K \right) \le \sum_{C(M,N)} \prod_{j=1}^{N}\mathbb{P}\left( Y_j > \frac{i_j K}{M} \right).
\end{align*}
\end{lemma}

\begin{proof}
$\forall (y_1,\cdots,y_N) \in \{(Y_1,\cdots,Y_N): \sum_{j=1}^{N} Y_{j}> K\}$, take
$
i_j = \left[ \frac{y_j M}{ \sum_{j=1}^{N} y_{j}}\right]$,  for $j=1,\ldots,N.
$
We can easily verify that $M-N \le \sum_{j=1}^{N} i_j \le M$ and $y_{j} > i_j K/M$. Therefore, we have
\begin{align*}
 \big\{(Y_1,\cdots,Y_N):\ \sum_{j=1}^{N} Y_{j}> K\big\}
\subset   \bigcup_{ C(M,N)} \left\{ (Y_1,\cdots,Y_N): \ Y_{j} > \frac{i_j K}{M} \mbox{ for }j=1,\ldots,N \right\}.
\end{align*}
Since $Y_{j}$'s are independent with each other, we obtain
\begin{align*}
\mathbb{P}\left(\sum_{j=1}^{N} Y_{j}> K \right) \le \sum_{C(M,N)} \prod_{j=1}^{N}\mathbb{P}\left( Y_{j} > \frac{i_j K}{M} \right).
\end{align*}
\end{proof}

\subsection{Proof of Theorem $1$}\label{proof-thm1}
First, we calculate the probability that our algorithm stops within time $p$. The value of $p$ is to be specified later.
\begin{align*}
\nonumber &\sum_{t=1}^{p} \mathbb{P}(T_s=t) 
=  \sum_{t=1}^{p} \left( \mathbb{P}\left(\{T_s=t\} \wedge B(\varepsilon,p)\right)+ \mathbb{P}\left(\{T_s=t\} \wedge \bar{B}(\varepsilon,p)\right) \right) \\
 \nonumber \le & \sum_{t=1}^{p} \mathbb{P}\left(\{T_s=t\}\wedge B(\varepsilon,p)\right)+ \mathbb{P}\left(\bar{B}(\varepsilon,p)\right) \\
 \le & \sum_{t=1}^{p} \mathbb{P}\left( \{T_s=t\} \wedge B(\varepsilon,p) \right)+ 2Np\times {\rm exp}\left(-\frac{(\varepsilon b - \mu_1)^2}{2\sigma_1^2}\right),
\end{align*}
where the last inequality is from Hoeffding Inequality and assumptions in Section 3. 
The value of $\varepsilon$ is to be specified later.

Denote $\bar{b} = N(1- \frac{  \sqrt{N}\varepsilon  \lambda_{2}}{1- \lambda_{2}})b$, then $\bar{b}$ will also tend to infinity as $b$ tends to infinity provided $\varepsilon$ small enough. By Lemma \ref{lemma_1}, we have
\begin{align}\label{eq7}
 \sum_{t=1}^{p} \mathbb{P}(T_s=t)  \le& \sum_{t=1}^{p} \mathbb{P}\left(\{\sum _{j=1}^{N} y_{j}^{t}> \bar{b}\} \wedge B(\varepsilon,p)\ \right) 
+ 2Np \times {\rm exp}\left(-\frac{(\varepsilon b - \mu_1)^2}{2\sigma_1^2}\right).
\end{align}
By Lemma \ref{lemma_2}, we have
\begin{align}\label{eq10}
 \mathbb{P}\left(\{\sum _{j=1}^{N} y_{j}^{t}> \bar{b}\}\wedge B(\varepsilon,p)\ \right)
\le  \sum_{C(M,N)} \prod_{j=1}^{N}\mathbb{P}\left( \{y_{j}^{t} > \frac{i_j \bar{b}}{M}\} \wedge B_j(\varepsilon,p) \right),
\end{align}
where $B_j(\varepsilon, p) = \{ |L({\bf x}_j^t)|< \varepsilon b,\mbox{ for }t=1,\ldots,p \}$ and the value of $M$ is to be specified later.
If $y_{j}^{t}>i_j \bar{b}/M$, then there must exist $1\le k\le t $ such that $y_{j}^{t} = \sum_{q=k}^{t} L({\bf x}_j^q) \ge i_j \bar{b}/M$. So, we have
\begin{align}\label{eq8}
 \mathbb{P}\left( \{y_{j}^{t} > \frac{i_j \bar{b}}{M}\} \wedge B_j(\varepsilon,p) \right) 
\le \sum_{k=1}^{t} \mathbb{P}\left( \{\sum_{q=k}^{t} L({\bf x}_j^q)> \frac{i_j \bar{b}}{M}\}\wedge B_j(\varepsilon,p) \right).
\end{align}
The influence of $B_j(\varepsilon,p)$ in \eqref{eq8} can be interpreted as truncating the original distribution of $L(\cdot)$.  It's obvious that the new distribution is still sub-Gaussian. Besides, the mean and variance almost keep unchanged provided $\varepsilon b$ large enough.

If $i_j=0$, we just set the upper bound of the probability in \eqref{eq8} to be $1$. If $i_j\neq 0$, by Lemma \ref{lemma_3}, we have
\begin{align}\label{eq9}
 \sum_{k=1}^{t} \mathbb{P}\left( \{\sum_{q=k}^{t} L({\bf x}_j^q) > \frac{i_j \bar{b}}{M}\}  \wedge B_j(\varepsilon,p) \right) 
  < -\frac{2i_j \bar{b}}{M\mu_1}  {\rm exp}\left(\frac{2i_j \bar{b}\mu_1}{M\sigma_1^{2}}\right).
\end{align}
Plugging \eqref{eq9} into \eqref{eq10}, we get
\begin{align}\label{eq11}
 &\sum_{C(M,N)} \prod_{j=1}^{N}\mathbb{P}\left(\{ y_{j}^{t} > \frac{i_j \bar{b}}{M} \} \wedge B_j(\varepsilon,p) \right)  <  \sum_{C(M,N)} \prod_{i_j\neq 0}  \frac{2i_j \bar{b}}{-M\mu_1}  {\rm exp}\left(\frac{2i_j \bar{b}\mu_1}{M\sigma_1^{2}}\right)\\
\nonumber  \le& \sum_{C(M,N)}  \left( \frac{2\bar{b}}{-\mu_1}\right)^N {\rm exp}\left(\frac{2\bar{b}\mu_1}{\sigma_1^{2}}(1-\frac{N}{M})\right)
 = |C(M,N)| \left( \frac{2\bar{b}}{-\mu_1}\right)^N {\rm exp}\left(\frac{2\bar{b}\mu_1}{\sigma_1^{2}}(1-\frac{N}{M})\right).
\end{align}
Plugging \eqref{eq11} and \eqref{eq10} into \eqref{eq7}, we obtain
\begin{equation}\label{eq91}
 \sum_{t=1}^{p} \mathbb{P}(T_s=t)  < \sum_{t=1}^{p} |C(M,N)| \left( \frac{2\bar{b}}{-\mu_1}\right)^N {\rm exp}\left(\frac{2\bar{b}\mu_1}{\sigma_1^{2}}(1-\frac{N}{M})\right) + 2Np \times {\rm exp}\left(-\frac{(\varepsilon b - \mu_1)^2}{2\sigma_1^2}\right)
 \end{equation}
 \begin{equation}
  =  p\left|C(M,N)\right| \left( \frac{2\bar{b}}{-\mu_1}\right)^N {\rm exp}\left(\frac{2\bar{b}\mu_1}{\sigma_1^{2}}(1-\frac{N}{M})\right) + 2Np\times  {\rm exp}\left(-\frac{(\varepsilon b - \mu_1)^2}{2\sigma_1^2}\right).
\end{equation}
Next, we will show that as $b$ tends to infinity, the second term on the RHS of  \eqref{eq91} is a small quantity in comparison with the first term if we choose the value of $M$ and $\varepsilon$ properly. Note that $2Np$ is a small quantity in comparison with $ p\left|C(M,N)\right| \left( -2\bar{b}/\mu_1\right)^N$, so we only require
$\frac{2\bar{b}\mu_1}{\sigma_1^{2}}(1-\frac{N}{M}) \ge -\frac{(\varepsilon b - \mu_1)^2}{2\sigma_1^2}.$
Choose $M=(N+1)^2$. Recall that $\bar{b} = N(1- \frac{  \sqrt{N}\varepsilon  \lambda_{2}}{1- \lambda_{2}})b$, the equation above can be rewritten as 
\begin{align}\label{eq92}
\frac{(\varepsilon b - \mu_1)^2}{2\sigma_1^2}
\ge -\frac{2(N^3+N^2+N)\mu_1 b}{(N+1)^2\sigma_1^2}\left(1- \frac{  \sqrt{N}\varepsilon  \lambda_{2}}{1- \lambda_{2}}\right).
\end{align}
To ensure that \eqref{eq92} holds as $b$ tends to infinity, $\varepsilon = 2\sqrt{-N\mu_1/b}$ is sufficient. 
Plugging the value of $M$ and $\varepsilon$ into \eqref{eq91} and neglecting the second term, we get
\begin{align*}
& \mathbb{P}\left(T_s \le p \right)\le  |C((N+1)^2,N)| \left( \frac{2\bar{b}}{-\mu_1}\right)^N   \\
& \cdot {\rm exp}\big(\frac{2(N^3+N^2+N)\mu_1 b}{(N+1)^2\sigma_1^{2}}-\frac{4N(N^3+N^2+N) \sqrt{-\mu_1}\mu_1\lambda_{2}\sqrt{b}}{(N+1)^2(1-\lambda_{2})\sigma_1^2}+\ln(p)\big).
\end{align*}
So $\forall l > -\frac{4N(N^3+N^2+N) \sqrt{-\mu_1}\mu_1\lambda_{2}}{(N+1)^2(1-\lambda_{2})\sigma_1^2}$,  if we choose
$
p = {\rm exp}\left(-\frac{2(N^3+N^2+N)\mu_1 b}{(N+1)^2\sigma_1^{2}} - l\sqrt{b}\right),
$
\begin{align*}
&  \lim_{b\rightarrow +\infty} \mathbb{P}\left(T_s \le p \right) \le \\
& \lim_{b\rightarrow +\infty} |C((N+1)^2,N)| \left( \frac{2\bar{b}}{-\mu_1}\right)^N   {\rm exp}\left(-\left(l+\frac{4N(N^3+N^2+N) \sqrt{-\mu_1}\mu_1\lambda_{2}}{(N+1)^2(1-\lambda_{2})\sigma_1^2}\right)\sqrt{b}\right)
=0,
\end{align*}
which together with the definition of ${\rm ARL}$ leads to
$
 {\rm ARL} \ge  p$, $\forall l > -\frac{4N(N^3+N^2+N) \sqrt{-\mu_1}\mu_1\lambda_{2}}{(N+1)^2(1-\lambda_{2})\sigma_1^2}$. 
This leads to our desired result 
when $b$ tends to infinity.
\subsection{Proof of Lemma $1$}\label{proof-lem1}
First of all, note that $\mathbb{P}\left(T_s =+\infty\right)=0$, so given  $\varepsilon>0$, we have
\begin{align}\label{eq00}
{\rm EDD} \le \frac{b(1+\varepsilon)}{\mu_2} + \sum_{t=[\frac{b(1+\varepsilon)}{\mu_2}]+1}^{+\infty}\mathbb{P}\left(T_s=t\right) t.
\end{align}
If $T_s=t$, then we have that $z_j^{t-1}<b$ holds for all $j$. Since $\sum_{j=1}^{N} z_{j}^{t-1} = \sum_{j=1}^{N} y_{j}^{t-1}$,  there must exist some $y_j^{t-1}<b$. Therefore, we have
\begin{align}\label{eq12}
\sum_{t=[\frac{b(1+\varepsilon)}{\mu_2}]+1}^{+\infty}\mathbb{P}\left(T_s=t\right) t 
\le  \sum_{t=[\frac{b(1+\varepsilon)}{\mu_2}]+1}^{+\infty} \sum_{j=1}^{N} \mathbb{P}\left( y_j^{t-1}<b \right) t.
\end{align}
Note that $y_j^{t-1}\ge\sum_{q=1}^{t-1}L({\bf x}_j^q) $, together with Hoeffding Inequality, we get
\begin{align}\label{eq13}
   \mathbb{P}\left( y_j^{t-1}<b \right)t \le& \mathbb{P}\left( \sum_{q=1}^{t-1}L({\bf x}_j^q)<b \right)
 = {\rm exp}\left(-\frac{1}{2}\left(\frac{b-(t-1)\mu_2}{\sqrt{t-1}\sigma_2}\right)^{2}\right)t.
\end{align}
When $b$ is large enough, for any $t>[\frac{b(1+\varepsilon)}{\mu_2}]$, utilizing the similar technique in \eqref{eq6}, we get
\begin{align}\label{eq14}
\frac{{\rm exp}\left(-\frac{1}{2}(\frac{b-t\mu_2}{\sqrt{t}\sigma_2})^{2}\right)}{ {\rm exp}\left(-\frac{1}{2}(\frac{b-(t-1)\mu_2}{\sqrt{t-1}\sigma_2})^{2}\right)} \times \frac{t+1}{t}
\le {\rm exp}\left( -\frac{b}{2}(1-\frac{1}{(1+\varepsilon)^2})\right).
\end{align}
Plugging \eqref{eq13} and \eqref{eq14} into \eqref{eq12}, utilizing the similar technique in  \eqref{eq71}, we get
\begin{align}\label{eq15}
 & \sum_{t=[\frac{b(1+\varepsilon)}{\mu_2}]+1}^{+\infty}\mathbb{P}\left(T_s=t\right) t
 \le  \sum_{t=[\frac{b(1+\varepsilon)}{\mu_2}]+1}^{+\infty} N\times {\rm exp}\left(-\frac{1}{2}\left(\frac{b-(t-1)\mu_2}{\sqrt{t-1}\sigma_2}\right)^{2}\right)t\\
 \le & N \left(\frac{b(1+\varepsilon)}{\mu_2}+1\right)
 \frac{{\rm exp}\left(-\frac{\varepsilon^2 b}{2(1+\varepsilon)\sigma_2^2}\right)}{1-{\rm exp}\left( -\frac{b}{2}(1-\frac{1}{(1+\varepsilon)^2})\right)}.
\end{align}
Note that $\forall \varepsilon >0$, as $b$ tends to infinity, the RHS of \eqref{eq15} would converge to zero. Therefore, by \eqref{eq00}, we get
$
{\rm EDD}\le \frac{b\big(1+o(1)\big)}{\mu_2}.
$

\clearpage


\clearpage

\end{document}